\documentclass[a4paper]{article}
\usepackage{graphicx,amssymb,amsmath,xcolor,enumerate}
\usepackage[lined,boxed,linesnumbered,noend,commentsnumbered,inoutnumbered]{algorithm2e}
\usepackage{url}

\newcommand\area{\mathsf{area}}
\newcommand\puz{\mathcal{P}}

\newtheorem{theorem}{Theorem}
\newtheorem{lemma}[theorem]{Lemma}
\newtheorem{conjecture}[theorem]{Conjecture}
\newtheorem{obs}{Observation}

\newenvironment{proof}{\emph{Proof.}}{\hfill $\Box$\\}

\newcommand{\theoremname}{Theorem}

\SetKwData{True}{TRUE}
\SetKwData{False}{FALSE}
\SetKwFunction{hasSA}{hasAssembly}
\SetKwFunction{three}{hasAssemblyCase3}
\SetKwFunction{issym}{isSym}
\SetKwFunction{join}{join}
\SetKwFunction{cong}{cong}
\SetKwProg{myfun}{Function}{}{}
\SetKwInOut{Input}{input}\SetKwInOut{Output}{output}

\title{Symmetric Assembly Puzzles are Hard, \\Beyond a Few Pieces}

\author{Erik D. Demaine\footnote{MIT, {\tt \{edemaine,jasonku\}@mit.edu}}
\and Matias Korman\footnote{Tufts University, {\tt matias.korman@tufts.edu}}
\and Jason S. Ku\footnotemark[1]
\and Joseph S. B. Mitchell\footnote{Stony Brook University, {\tt joseph.mitchell@stonybrook.edu}}
\and Yota Otachi\footnote{Kumamoto University, {\tt otachi@cs.kumamoto-u.ac.jp}}
\and Andr\'e van Renssen\footnote{University of Sydney, {\tt andre.vanrenssen@sydney.edu.au}}
\and Marcel Roeloffzen\footnote{TU Eindhoven, {\tt m.j.m.roeloffzen@tue.nl}}
\and Ryuhei Uehara\footnote{JAIST, {\tt uehara@jaist.ac.jp}}
\and Yushi Uno\footnote{Osaka Prefecture University, {\tt uno@cs.osakafu-u.ac.jp}}
}

\date{}

\begin{document}

\maketitle

\begin{abstract}
We study the complexity of symmetric assembly puzzles: given a collection of
simple polygons, can we translate, rotate, and possibly flip them so that their
interior-disjoint union is line symmetric? On the negative side, we show that
the problem is strongly NP-complete even if the pieces are all polyominos. On
the positive side, we show that the problem can be solved in polynomial time if
the number of pieces is a fixed constant.
\end{abstract}

\section{Introduction}

The goal of a 2D \emph{assembly puzzle} is to arrange a given set of pieces so 
that they do not overlap and form a target silhouette. The most famous
example is the Tangram puzzle, shown in \figurename~\ref{fig:tangram}. Its
earliest printed reference is from 1813 in China, but by whom or exactly when it
was invented remains a mystery~\cite{Slocum2004}. There are over 2,000 Tangram
assembly puzzles~\cite{Slocum2004}, and many more similar 2D assembly
puzzles~\cite{FoxEpsteinUehara2014}. A recent trend in the puzzle world is a
relatively new type of 2D assembly puzzle that we call \emph{symmetric assembly
puzzles}. In these puzzles the target shape is not specified. Instead, the
objective is to arrange the pieces so that they form a symmetric silhouette
without overlap.

\begin{figure}[htb]
  \centering
  \includegraphics[width=0.45\linewidth]{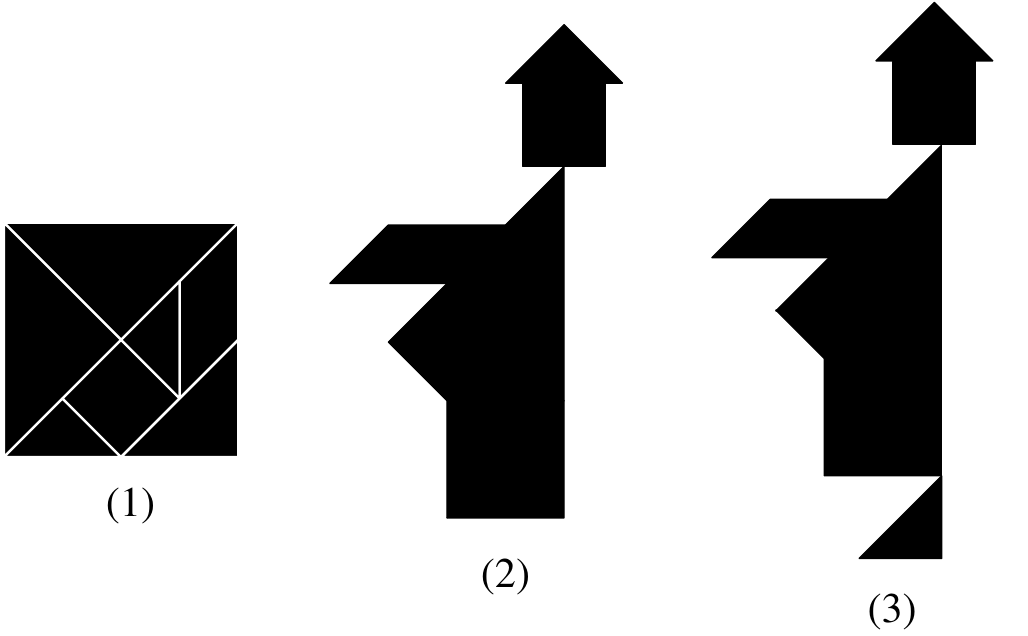}\qquad
  \includegraphics[width=0.45\linewidth]{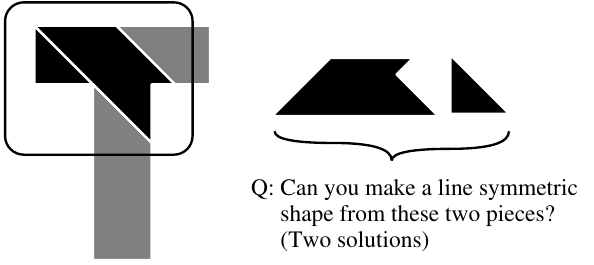}
  \caption{[Left] The seven Tangram pieces (1) can be assembled into non-simple silhouettes (2) and (3). 
  [Right] A symmetric assembly puzzle invented by Hiroshi Yamamoto~\cite{Yamamoto2014}: 
  given the two black pieces (right) from the classic T puzzle (left), 
  make two different line symmetric shapes.  (Used with permission.)}
 \label{fig:tangram}
\end{figure}

The first symmetric assembly puzzle, ``Symmetrix'', was designed in 2003 by
Japanese puzzle designer Tadao Kitazawa and was distributed by Naoyuki Iwase
as his exchange
puzzle at the 2004 International Puzzle Party (IPP) in Tokyo~\cite{IPP2004}.
The lack of a specified target shape makes these puzzles quite
difficult to solve.
In this paper, we aim for arrangements that are line symmetric
(reflection through a line),
but other symmetries such as rotational symmetry could also be considered.
We also assume that the given pieces are simple polygons, and that
the line-symmetric assembly must be a simple polygon (have no holes).

We study the computational complexity of this general form of
symmetric assembly puzzle.
Precisely, we define a \emph{symmetric assembly puzzle}, or SAP, to be a set
of $k$ disjoint simple polygons $\puz = \{P_1,P_2,\ldots,P_k\}$, with
$n=|P_1|+\cdots+|P_k|$ the total number of vertices in all pieces. By
\emph{simple polygon} we mean a closed subset of $\mathbb{R}^2$ homeomorphic to
a disk bounded by a closed path of a finite number of straight line segments
where nonadjacent edges and vertices do not intersect. A \emph{symmetric
assembly} $f: \{ p \in P \mid P \in \puz \} \rightarrow \mathbb{R}^2$, of a SAP
$\puz$ is a planar isometric embedding of the pieces ($\{ f(p) \mid p \in P \}$
for each $P \in \puz$ is a rigid transformation of $P$), such that their mapped
interiors are disjoint and their mapped union forms a simple polygon that is
line symmetric.
We abuse notation slightly by using $f(P)$ to denote $\{ f(p) \mid p \in P \}$
and $f(\puz)$ to denote $\{ f(p) \mid p \in P,~ P \in \puz \}$. 
We refer to SAP (symmetric assembly puzzles) as the problem of deciding whether
an instance $\puz$ has a symmetric assembly $f$, and we study the computational
complexity of SAP.
We allow pieces to flip over (reflect), but other variants of the puzzle may
disallow this. Given that humans have difficulty solving SAPs with even a few
low-complexity pieces, we consider two different generalizations: bounded piece
complexity ($|P_i| = O(1)$) and bounded piece number ($k = O(1)$).  In the
former case, we prove strong NP-completeness, while in the latter case, we solve
the problem in polynomial time (the exponent is linear in~$k$).

\section{Many Pieces}

First, we show that it is hard to solve symmetric assembly puzzles with a large number of pieces,
even if each piece has bounded complexity ($|P_i|=O(1)$).

\begin{theorem}
\label{th:hard}
Symmetric assembly puzzles are strongly NP-complete even if each piece is a
polyomino with at most six vertices with area upper-bounded by a
polynomial function of the number of pieces.
\end{theorem}

\begin{proof}
If a SAP has a solution, the location and orientation
of each piece within a symmetric assembly is a solution
certificate of polynomial size checkable in polynomial time,
so symmetric assembly puzzles are in NP.
We reduce from the \textsc{Rectangle Packing Puzzle} problem (in short the RPP problem), 
known to be strongly NP-hard~\cite{DemaineD07}. Specifically,
it is (strongly) NP-complete to decide whether $k$ given rectangular 
pieces---sized $1 \times x_{1}, 1 \times x_{2}, \dots, 1 \times x_{k}$,
where the $x_{i}$'s are positive integers bounded above by a 
polynomial in $k$---can be exactly packed into a specified rectangular box 
with given width $w$ and height $h$ and area $x_{1} + x_{2} + \dots + x_{k}=wh$.

\begin{figure}[htb]
  \centering
  \includegraphics[width=0.4\linewidth]{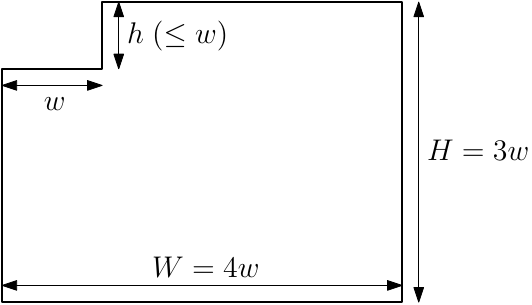}
  \includegraphics[width=0.3\linewidth]{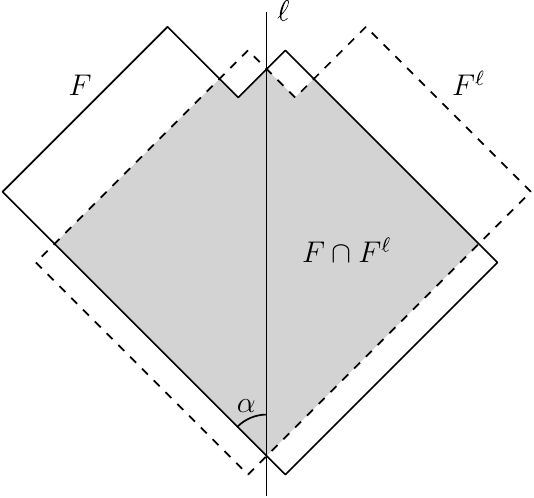}
  \includegraphics[width=0.2\linewidth]{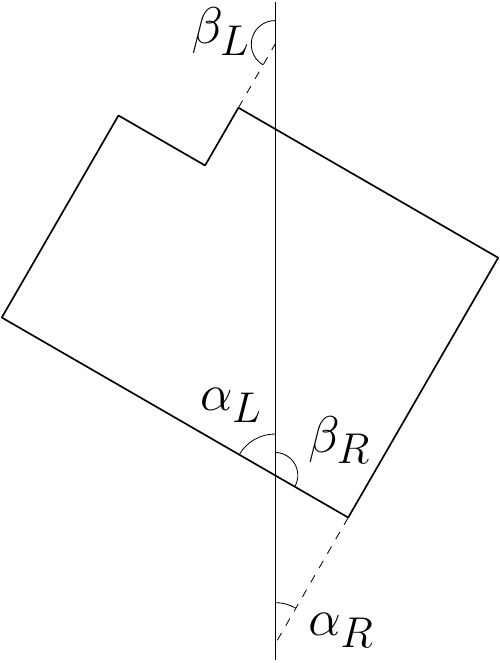}
  \caption{[Left] The frame piece $F$. 
  [Middle] If $\ell$ and $\ell_B$ form an angle of $\pi/4$, then $F \cap F^{\ell}$ is contained 
  in a rectangle in an $H \times H$ and thus $O^*$ cannot be line symmetric.
  [Right] The angles $\alpha_{L}$, $\beta_{L}$, $\alpha_{R}$, and $\beta_{R}$.}
  \label{fig:frame}
\end{figure}

Let $I = (x_{1}, \dots, x_{k}, w, h)$ be a rectangle packing puzzle.
Without loss of generality, we assume that $w \ge h$.
Now let $I' = (P_{1}, \dots, P_{k}, F)$ be the SAP where
$P_{i}$ is the $1 \times x_{i}$ rectangle for each $i\in\{1,\ldots,k\}$,
and $F$ is the polyomino in \figurename~\ref{fig:frame}. We call $F$ the {\em frame piece} of $I'$. 
We show that $I$ has a rectangle packing if and only if
$I'$ has a symmetric assembly.

Clearly, if $I$ has a rectangle packing, then the pieces $P_{1}, \dots, P_{k}$
can be packed into the $w \times h$ hole in the frame piece
creating a line symmetric $W \times H$ rectangle, solving the
SAP. Now we show the reverse implication. Assume that $I'$
has a symmetric assembly, and let $O^*$ be a line symmetric polygon formed by
the pieces $\{P_{1}, \dots, P_{k}, F\}$. We claim that $O^*$ must be a $W\times
H$ rectangle, which will imply that $I$ is a yes-instance of RPP.  Fix a
placement of the pieces of $I'$ that forms $O^*$, and let $\ell$ be one of its
lines of symmetry. Assume, without loss of generality, that $\ell$ is a vertical
line. Let $F^{\ell}$ be the reflection of $F$ about $\ell$. 

\begin{obs}
  \label{obs:intersection}
$\area(F \cap F^{\ell}) \ge W H - 2 w h \ge 10w^{2}$
\end{obs}
\begin{proof}
Because $O^*$ contains $F^{\ell}$ and $F$,
it holds that $\area(F^{\ell} \setminus F) \le \area(O^* \setminus F) = w h$.
Because $F \cup F^{\ell}$ is mirror-symmetric, 
$\area(F^{\ell} \setminus F) = \area(F \setminus F^{\ell})$.
Hence, it follows that $\area(F \cap F^{\ell}) = \area(F) - \area(F\setminus F^{\ell}) \ge W H - 2 w h \ge 10w^{2}$.
\end{proof}

Observation~\ref{obs:intersection} implies that $\ell$ passes through an interior point of $F$.  
Let $\ell_B$ be the line containing the segment of $F$ with length $4w$. 
Let $c$ be the center of the frame piece's bounding box.

\begin{lemma}
  \label{lem:alpha1}
  $\ell_B$ is either parallel to $\ell$ or orthogonal to $\ell$.
 \end{lemma}

\begin{proof}
Suppose, for contradiction, that $\ell_B$ is neither parallel nor orthogonal to
$\ell$. Let $\alpha$ be the smaller angle made by $\ell_B$ and $\ell$.\@ We
partition the edges of $F$ crossed by $\ell$ into two at their intersection
points. Let $F_{L}$ and $F_{R}$ be the sets of segments on the left and right
portions of $F$, respectively.
Consider the set of counter-clockwise angles between $\ell$ and the lines
containing segments of $F_{L}$. The assumptions that $\ell_B$ and $\ell$ are
neither parallel nor orthogonal, and that $F$ is a polyomino together imply that
the set contains exactly two angles $\alpha_{L}$ and $\beta_{L}$, where
$\alpha_{L} \le \beta_{L}$ and $\alpha_{L} + \pi/2 = \beta_{L}$.
Similarly, let $\alpha_{R}$ and $\beta_{R}$ be the clockwise angles between
$\ell$ and the lines containing segments of $F_{R}$, where $\alpha_{R} \le
\beta_{R}$ and $\alpha_{R} + \pi/2 = \beta_{R}$. Because $\alpha_{L} + \beta_{R} =
\pi$, it holds that $\alpha_{L} + \alpha_{R} = \pi/2$. Note that $\alpha \in
\{\alpha_{L}, \alpha_{R}\}$.

Two distinct pieces of $I'$ are \emph{connected} if the fixed placement of
the two pieces to form $O^*$ have a non-degenerate line segment on their edges in
common. Let $\puz$ be the subset of $\{P_{1}, \dots, P_{k}, F\}$ such that each
$P_{i} \in \puz$ can be reached from $F$ by repeatedly following connected pieces
in $O^*$.

As before, consider the angles formed by $\ell$ and the lines containing segments in the left and right parts of $\puz$.
Because all pieces are polyominoes, these lines cannot make angles 
other than $\alpha_{L}$, $\beta_{L}$, $\alpha_{R}$, and $\beta_{R}$ with $\ell$.
Further note that the subset $O'$ of $O^*$ covered by $\puz$ must be mirror-symmetric with respect to $\ell$, 
or else $O^*$ would not be.
This implies that $\alpha_{L} = \alpha_{R}$.
Because $\alpha_{L} + \alpha_{R} = \pi/2$, the only solution in which $\ell$ is not parallel or orthogonal to $\ell_B$ is when $\alpha_{L} = \alpha_{R} = \pi/4$ and $\alpha = \pi/4$.
However, if $\alpha = \pi/4$, then $F \cap F^{\ell}$ is a subset of an 
$H \times H$ rectangle (see \figurename~\ref{fig:frame}),
whose area is at most $H^{2} = 9w^{2}$, contradicting Observation~\ref{obs:intersection}.
\end{proof}

\begin{figure}
  \centering
  \includegraphics[width=0.55\linewidth]{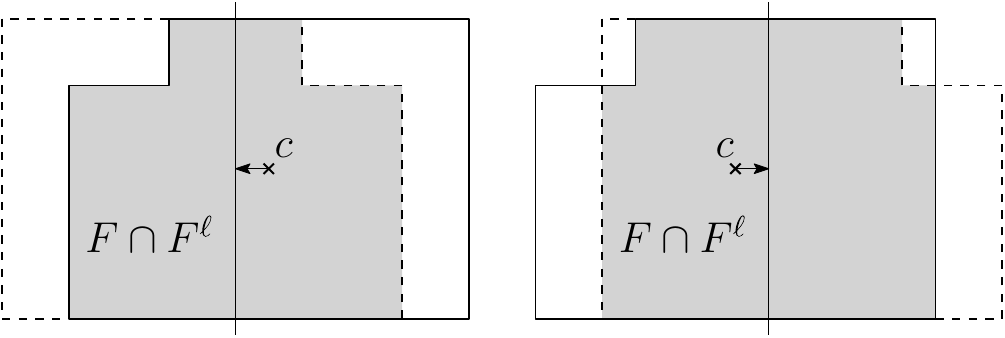}\qquad
  \includegraphics[width=0.35\linewidth]{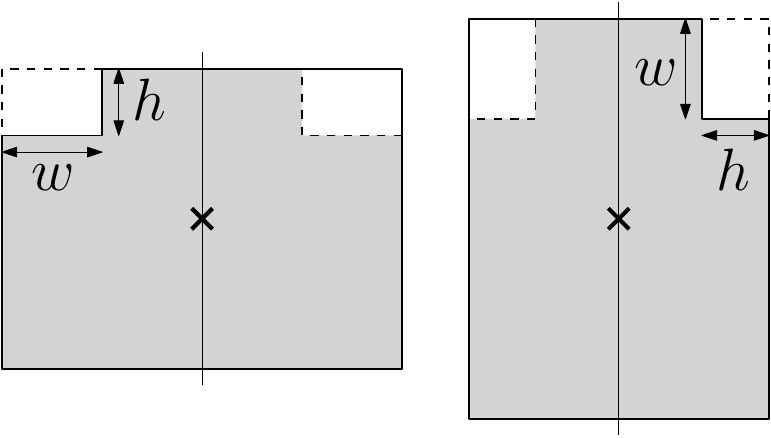}
  \caption{[Left] When $\ell$ passes to the left of $c$, the portion of $F$ to the left of $\ell$ is too small.
  If it passes to the right, the right portion would be too small. 
  [Right] If $\ell$ passes through $c$, and is either orthogonal or 
  parallel to $\ell_B$, the symmetric assembly puzzle 
  can only be completed into a rectangle.}
  \label{fig:center}
\end{figure}

So $\ell$ is either parallel or orthogonal to $\ell_B$. Further, it passes
through $c$ (see \figurename~\ref{fig:center}). In either case, $F \cup F^{\ell}$
is a $W \times H$ rectangle, and thus $O^* = F \cup F^{\ell}$. This implies that
$O^* \setminus F$ is a $w \times h$ rectangle that must contain the remaining
pieces of $I'$. In particular, we have that this placement 
packing of $P_{1}, \dots, P_{k}$ gives a solution to the instance $I$ of RPP,
completing the proof of \theoremname~\ref{th:hard}.
\end{proof}

We extend the above proof to show that the problem remains strongly NP-complete even when each piece is a convex quadrilateral. 

\begin{theorem}
  \label{theo:ConvexNPComplete}
  Symmetric assembly puzzles are strongly NP-complete even if each piece is a convex quadrilateral and area upper bounded by a polynomial function of the number of pieces. 
\end{theorem}
\begin{proof}

We note that the only piece that is not a convex quadrilateral is the frame
piece $F$. Hence, we split this into two convex quadrilateral pieces as shown in
\figurename~\ref{fig:SplittingFrame}. We note that due to the dimensions of $H$
and $W$, the four angles $\alpha$, $\beta$, $\gamma$, and $\delta$ are all
unique. Furthermore, only $\alpha + \delta$ and $\beta + \gamma$ do not sum up
to multiples of $\pi/2$. If we show that any line symmetric solution aligns these
four angles as in \figurename~\ref{fig:SplittingFrame},
Theorem~\ref{th:hard} completes the proof. 

\begin{figure}
  \begin{minipage}[t]{0.5\linewidth}
    \begin{center}
      \includegraphics{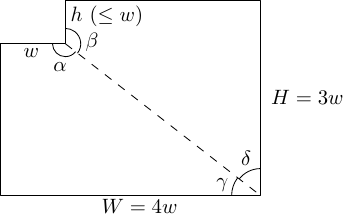}
    \end{center}
    \caption{Splitting the frame piece into two convex quadrilateral pieces.}
    \label{fig:SplittingFrame}
  \end{minipage}
  \hspace{0.05\linewidth}
  \begin{minipage}[t]{0.45\linewidth}
    \begin{center}
      \includegraphics{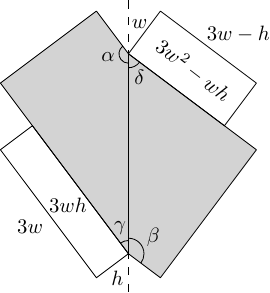}
    \end{center}
    \caption{Matching $\alpha$ to $\delta$ and $\beta$ to $\gamma$.}
    \label{fig:WrongMatching}
  \end{minipage}
\end{figure}

Assume the angles are not matched as in \figurename~\ref{fig:SplittingFrame}. We
first show that extending $\gamma$ or $\delta$ by a multiple of $\pi/2$ is not
useful. We focus on $\gamma$, but the same argument holds for $\delta$. If we
extend $\gamma$ using a right angle of the other frame piece, it creates an
imbalance resulting from the implied line of symmetry cannot be overcome using
only the remaining rectangles of combined area $wh$ (see
\figurename~\ref{fig:ExtendingAngle}). Extending $\gamma$ using the rectangles
also does not lead to a line symmetric polygon, because placing the other frame
piece afterwards still leads to an imbalanced shape.

\begin{figure}
  \begin{center}
    \includegraphics{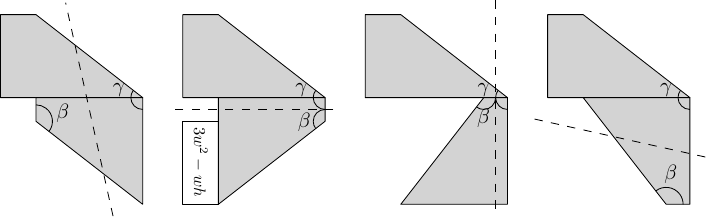}
  \end{center}
  \caption{The four cases when extending $\gamma$ using the other frame piece.}
  \label{fig:ExtendingAngle}
\end{figure}

Because the four angles are all unique and the symmetry line can pass through at most two corners of a simple polygon, at least two of these angles have to meet in a point. If $\alpha$ is matched to $\delta$ or if $\beta$ is matched to $\gamma$, we note that the created angle is not a multiple of $\pi/2$ and thus we still have three unique angles. This implies that in this case, both $\alpha$ is matched to $\delta$ and $\beta$ is matched to $\gamma$ (see \figurename~\ref{fig:WrongMatching}). 

We first show that the difference between $\alpha + \delta$ and $\beta + \gamma$ cannot be a multiple of $\pi/2$, which implies that the line of symmetry still needs to pass through both of these angles. We prove this by contradiction, so assume that the difference is a multiple of $\pi/2$. We observe from \figurename~\ref{fig:SplittingFrame} that $\alpha + \delta + \beta + \gamma = 2 \pi$, $\beta = \gamma + \pi/2$, $\alpha = \delta + \pi/2$, and $\gamma < \delta$. Hence, we need to consider only the case where $\alpha + \delta = \beta + \gamma + \pi/2$, which implies that $\delta = \gamma + \pi/4$. Because $\delta + \gamma = \pi/2$, it follows that $\gamma = \pi/8$ and that $\tan \gamma = \sqrt{2} - 1$. However, from \figurename~\ref{fig:SplittingFrame} we also observe that $\tan \gamma = (H-h)/3w$, where $2w \leq H-h < 3w$, which implies that $\tan \gamma \geq 2/3$, contradicting that $\gamma$ is $\pi/8$. Thus, the difference between $\alpha + \delta$ and $\beta + \gamma$ also cannot be a multiple of $\pi/2$. 

Because neither $\alpha + \delta$ nor $\beta + \gamma$ is a multiple of $\pi/2$ and the difference between $\alpha + \delta$ and $\beta + \gamma$ also cannot be a multiple of $\pi/2$, the only way to construct a line symmetric solution is for the symmetry line to pass through both created angles. However, this implies that in order to make a line symmetric shape, we need to at least add one region of area $3w^2 - wh$ and one of area $3wh$. Hence, the total area required is at least $3w^2 + 2wh$, which is more than the $wh$ combined area of the rectangles. Therefore, the four angles have to be aligned as in \figurename~\ref{fig:SplittingFrame}. 
\end{proof}

This result raises the question of what the simplest shape is for which the problem is strongly NP-complete. We conjecture that the problem is still strongly NP-complete even when each piece is a right triangle. 

\begin{conjecture}
  \label{conj:TriangleNPComplete}
  Symmetric assembly puzzles are strongly NP-complete even if each piece is a right triangle with area upper-bounded by a polynomial function of the number of pieces. 
\end{conjecture}

While we do not have a proof of this conjecture, we do sketch an approach to a possible proof based on a reduction from the \textsc{3-Partition} problem: It is (strongly) NP-complete to decide whether a given set of $3k$ positive integers (each integer is bounded from above by a polynomial in $k$) can be partitioned into $k$ triples, such the sum of the integers in each triple is the same. 

Let $\{a_1, ..., a_{3k}\}$ be the given set of integers in increasing order. We first transform these integers into almost squares of size $1 \times 1 + \epsilon_i$, such that the $1 + \epsilon_i$ sides of each triple sum to the same length: When we want to ensure that $\epsilon_i$ is at most $1/1000$ for each square, we transform each $a_i$ into an almost square of size $1 \times 1 + \frac{a_i}{1000 a_{3k}}$. Note that this does not change triples nor the solvability of the \textsc{3-Partition} instance. 

\begin{figure}
  \begin{minipage}[t]{0.45\linewidth}
    \begin{center}
      \includegraphics{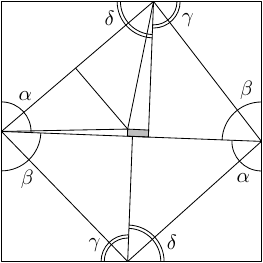}
    \end{center}
    \caption{The frame and its splitting lines. The hole for the \textsc{3-Partition} instance is shown in gray.}
    \label{fig:FrameTriangle}
  \end{minipage}
  \hspace{0.05\linewidth}
  \begin{minipage}[t]{0.45\linewidth}
    \begin{center}
      \includegraphics{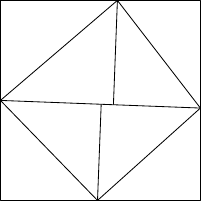}
    \end{center}
    \caption{Splitting an almost square.}
    \label{fig:AlmostSquare}
  \end{minipage}
\end{figure}

Next, we create a big square frame that has a hole of size $k \times 3 + \frac{\sum_{i=1}^{3k} a_i / k}{1000 a_{3k}}$. Note that the area of this hole is equal to the total area of the almost squares. We split the frame into right triangles as shown in \figurename~\ref{fig:FrameTriangle}, while ensuring that any combination of non-right angles is unique. 

Finally, we split the $3k$ almost squares into $24k$ right triangles. The general idea behind the splits is the same as for the frame: for each almost square, we pick four points close to the middle of its sides and split the square as shown in \figurename~\ref{fig:AlmostSquare}. More precisely, when $s$ is the length of a side, we pick a point $p \in \{\frac{s}{2} + \frac{is}{2k^{20}} \colon i\in \{1, 2, \ldots, k^{19}\}\}$and split along the line connecting the two points on the vertical sides, along the lines from the points on the horizontal sides perpendicular to the previous splitting line, and along the lines defined by points on consecutive sides (see \figurename~\ref{fig:AlmostSquare}). Note that $p$ is at most $s/2k$ away from the middle of the side. Again, we require that any combination of non-right angles is unique. 

This uniqueness of angles should ensure that the triangles can only be combined to the desired frame and almost squares. Proving this formally, however, turns out to be rather intricate; thus, we leave the full proof of the above conjecture for future research.

\section{Constant Number of Pieces}

Next we analyze symmetric assembly puzzles with a constant number of
pieces but many vertices, and show they can be solved in
polynomial time.

\begin{theorem}
\label{th:constpieces}
Given a symmetric assembly puzzle
with a constant number of pieces $k$
containing at most $n$ vertices in total,
deciding whether it has a symmetric assembly
can be decided in time that is polynomial in~$n$.
\end{theorem}

To prove this theorem, we present a brute force algorithm for solving a SAP that
runs in polynomial time for constant $k$. We say two pieces in a symmetric
assembly are \emph{connected} to each other if their intersection in the
symmetric assembly contains a non-degenerate line segment, and let the
\emph{connection} between two connected pieces be their intersection not
including isolated points. We will call two pieces \emph{fully} connected if
their connection is exactly an edge of one of the pieces, and \emph{partially}
connected otherwise; note that two pieces may be partially connected along more
than one edge. Call a piece a \emph{leaf} if it connects to at most one piece,
and a \emph{branch} otherwise. Given a leaf, let its \emph{parent} be the piece
connected to it (if it exists), and let its \emph{siblings} be all other pieces
connected to its parent. An illustration demonstrating these terms can be found
in \figurename~\ref{fig:def}.

In addition, we will need to construct simple polygons from provided simple
polygons by laying them next to each other along an edge. Let $E_P$ denote the
set of directed edges $(p_i,p_j)$ from a vertex $p_i$ to an adjacent vertex
$p_j$ of some simple polygon $P$.

Given an edge $e\in E_P$, we denote its length by $\lambda(e)$. Let $e_P =
(p_1,p_2)$ be a directed edge of a polygon $P$, let $e_Q =(q_1,q_2)$ be a
directed edge of a polygon $Q$, and let $d$ be a (non-negative) length strictly
less than $\lambda(e_P)$. Orient $P$ and $Q$ such that $e_P$ exists in a
clockwise traversal of $P$, $e_Q$ exists in a counter clockwise traversal of
$Q$, $e_Q$ is collinear and in the same direction as $e_P$, and the distance
between $p_1$ and $q_1$ is distance $d$. Call these transformations the mapping
$g: \{ p \in P \cup Q \mid P, Q \in \puz \} \rightarrow \mathbb{R}^2$, where $\{
g(p) \mid p \in P \cup Q \}$ for each $P, Q \in \puz$ is a rigid transformation
of $P$ and $Q$. 
Then we define \join{$e_P,e_Q,d$} to be $g(P)\cup g(Q)$ when $g(P)\cup g(Q)$ is
a simple polygon, and otherwise the empty set. See \figurename~\ref{fig:def}.

\begin{figure}[htb]
  \centering
  \includegraphics[width=\textwidth]{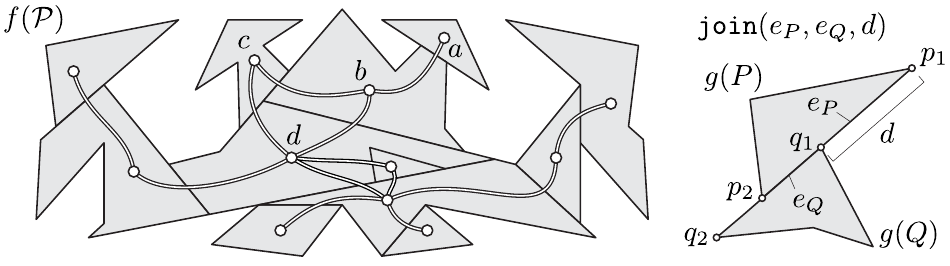}     
  \caption{ 
  [Left] Example symmetric assembly $\puz$ showing its connection graph. 
  Pieces $a$ and $d$ are fully connected to piece $b$, and $c$ is partially connected to $b$. 
  Pieces $b$, $c$, and $d$ are branches. Piece $a$ is a leaf, with
  $b$ its parent and $c$ and $d$ the siblings of $a$.
  [Right] Visualization of a \textsc{join} operation.}
  \label{fig:def}
\end{figure}

If a SAP has a symmetric assembly, let its \emph{connection graph} be a graph on
the pieces with an edge connecting two pieces if they are connected in the
symmetric assembly. Because a symmetric assembly is a simple polygon by
definition, its connection graph is connected and has a spanning tree; we can
then construct the assembly using a concatenation of \join procedures in
breadth-first-search order from an arbitrary root. Because parameter $d$ is not
discrete, the total solution space of simple polygons that are constructible
from the pieces of a SAP may be uncountable. However, we will exploit the
structure of symmetric assemblies to search only a finite set of configurations.

\begin{figure}[htb]
  \centering
  \includegraphics[width=\textwidth]{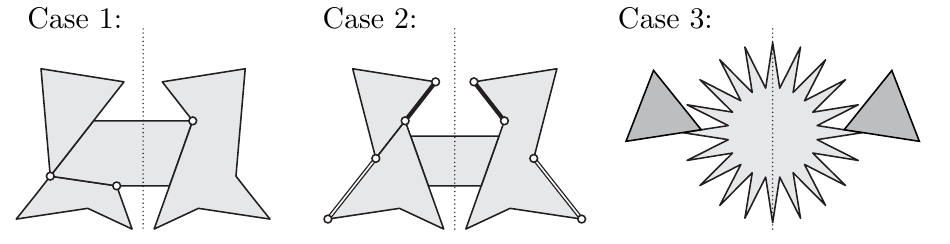}
  \caption{Examples of symmetric assemblies belonging to each case.
  Case 1 highlights vertices of connected pieces that intersect. Case 2
  highlights \textsc{join} operations using lengths of piece edges. Case 3
  is constructed from one symmetric piece
  and a pair of congruent pieces. }
  \label{fig:cases}
\end{figure}

In order to enumerate possible configurations, 
we would like to distinguish between three cases of the puzzle (see \figurename~\ref{fig:cases}), specifically:
\begin{enumerate}[{Case} 1:]
\item the puzzle has a symmetric assembly in which two 
connected pieces share a vertex on their connection;
\item the puzzle has no symmetric assembly satisfying Case~1, but has one in 
which the distance between vertices from the connecting edges between two
connected pieces has the same length as another whole edge in the piece set
(we say the connection between the two pieces \emph{constructs} the length of
another edge); 
\item the puzzle has no symmetric assembly satisfying Case~1 or Case~2, but has
one in which a nonempty set of pieces are themselves line symmetric about the
line of symmetry of the symmetric assembly, and any remaining pieces are pairs
of congruent pieces symmetric about the line of symmetry.

\end{enumerate}

The following lemma ensures that these three cases are exhaustive.

\begin{lemma}
If a SAP has a symmetric assembly, it can be described by one of the above three
cases.
\label{lem:cases}
\end{lemma}

To prove this lemma, we will use the following auxiliary results.

\begin{lemma}
If a SAP has a symmetric assembly that is not Case~1, the connection
graph of the symmetric assembly is a tree and all connections are single line
segments.
\label{lem:tree}
\end{lemma}

\begin{proof}

Let $\puz$ be a SAP with symmetric assembly $f: \{ p \in P \mid P \in \puz \} \rightarrow \mathbb{R}^2$, such that $\{ f(p) \mid p \in P \}$ for $P \in \puz$ is a rigid transformation of $P$,
that is not Case~1. Suppose, for contradiction, that the connection graph of $f(\puz)$ is not a tree, so that there
exists a cycle $C$ in the connection graph. Let $S$ be a simple closed curve
embedded in $f(\puz)$ that traverses the piece connections from $C$. The region
$R$ bounded by $S$ is completely covered by $f(\puz)$, or else it would contain
a hole, contradicting that $f(\puz)$ is a simple polygon. Because $f(\puz)$ covers
$R$ and $S$ corresponds to a cycle in the connection graph, then $R$ contains a
vertex $v$ in $R$ of some piece $P$. But then $P$ must share vertex $v$ along
its connection with another piece, contradicting exclusion from Case~1. 

So the connection graph of $f(\puz)$ is a tree. Now suppose for contradiction
there exists two connected pieces $P$ and $Q$ whose connection is more than one
line segment. Then there exists a closed curve embedded in $f(\puz)$ that
crosses from $P$ to $Q$ along two distinct edges. Then the region $R$ bounded
by the cycle must contain a vertex $v$ of $P$. If $f(\puz)$ covers $R$, $P$ must
share vertex $v$ along its connection with a vertex of $Q$, contradicting
exclusion from Case~1. Otherwise, $f(\puz)$ does not cover $R$, contradicting
that $f(\puz)$ is a simple polygon. \end{proof}

\begin{lemma}
If a SAP has a symmetric assembly that is not Case~1 or Case~2, the reflection
of any partially connected leaf is exactly another piece congruent to the leaf.
\label{lem:partial}
\end{lemma}

\begin{proof}

Let $\puz$ be a SAP with a symmetric assembly $f: \{ p \in P \mid P \in \puz \} \rightarrow \mathbb{R}^2$, such that $\{ f(p) \mid p \in P \}$ for $P \in \puz$ is a rigid transformation of $P$,  
that is not Case~1 or Case~2. 
Let $s: f(\puz)\rightarrow f(\puz)$ be an automorphism reflecting $f(\puz)$ across a line of symmetry $L$, 
and let $\mu = s \circ f$, mapping each point $p\in P_i$ of a piece $P_i\in \puz$ to the corresponding point in the reflection of $f(P_i)$ across $L$.

Consider a partially connected leaf $P$ whose parent is $Q$ with edge $e_P$
connected to edge $e_Q$, and suppose for contradiction that $\mu(P)$ is not
exactly covered by another piece congruent to $P$. We first show that a single
piece $P'$ contains $\mu(e_P)$ under $f$ so that $f(P')\subset \mu(P)$, 
and then show that in fact $f(P') = \mu(P)$.

By Lemma~\ref{lem:tree} the partial connection is a single line segment, and
$\ell_P = f(e_P)\setminus f(e_Q)$ is non-empty. $s(\ell_P)$ cannot be covered by
more than one piece or else two pieces would share a vertex along their
connection contradicting exclusion from Case~1. Also $s(\ell_P)$ cannot be
exactly the edge of another piece or else the connection between $P$ and $Q$
would construct its length, contradicting exclusion from Case~2. Thus,
$s(\ell_P)$ is a strict subset of an edge $e_{P'}$ from some piece $P'$ under
$f$. Further, $\mu(e_P) = f(e_{P'})$. Suppose for contradiction it did not, and
an endpoint $p$ of $e_{P'}$ maps to a point interior to $f(e_P)\cap f(e_Q)$. Then
$f(p)$ is also an interior point of $f(\puz)$, so $\mu(p)$ is also an interior
point, and $f(e_{P'})$ would share a vertex along its connection with another
piece contradicting exclusion from Case~1. So $\mu(e_P)$ is exactly edge
$e_{P'}$ of $f(P')$. And because $P$ is a leaf, $f(P') \subseteq \mu(P)$.

Now we show that in fact $f(P') = \mu(P)$. Suppose for contradiction that
$f(P')$ is a strict subset of $\mu(P)$, meaning that some other piece is also
fully contained in $\mu(P)$. Let $Q'$ be the first such piece connecting to $P'$
in a clockwise traversal of $P'$ from $e_{P'}$. Then the connection between $Q'$
and $P'$ must either construct the length of some edge from $P$ under $f$,
contradicting exclusion from Case~2, or $Q'$ and $P'$ must share a vertex on
their connection, contradicting exclusion from Case 1. So, $P'$ is a piece
congruent to $P$. 
\end{proof}

We now use these intermediate results to prove Lemma~\ref{lem:cases}.

\begin{proof}[of Lemma~\ref{lem:cases}]

Suppose for contradiction there exists a SAP $\puz$ having a symmetric assembly
$f: \{ p \in P \mid P \in \puz \} \rightarrow \mathbb{R}^2$, such that $\{ f(p) \mid p \in P \}$ for $P \in \puz$ is a rigid transformation of $P$,  
that does not satisfy any of the above cases,
and assume $\puz$ has the fewest pieces among all such SAPs. We will identify a
symmetric leaf or a congruent pair of leaves that can be removed from $\puz$ to
form a SAP with fewer pieces. The new SAP must have the same classification as
the original, contradicting the minimality of $\puz$.

Let $s: f(\puz)\rightarrow f(\puz)$ be an automorphism reflecting $f(\puz)$ across a line of symmetry $L$, 
and let $\mu = s \circ f$, mapping each point $p\in P_i$ of a piece $P_i\in \puz$ to the corresponding point in the reflection of $f(P_i)$ across $L$.
Let $P$ be a leaf in the
symmetric assembly whose siblings include at most one branch. Such a $P$ exists,
as any leaf with longest distance from an arbitrary root satisfies this
property. We claim that either $P$ is symmetric about line of symmetry $L$, or
$\mu(P)$ is exactly covered by a second piece of the SAP congruent to $P$,
contradicting the minimality of $\puz$.

First, if $P$ has no parent and is the only piece in the symmetric assembly, $P$
must be a line symmetric polygon. Otherwise, let $Q$ be the parent of $P$ with
edge $e_P$ of $P$ connected to edge $e_Q$ of $Q$. Let $e_{PQ}$ denote the subset
of $e_Q$ that maps to the intersection $f(e_P)\cap f(e_Q)$. We show that
$f(e_{PQ})$ and $\mu(e_{PQ})$ are not the same line segment. Suppose for
contradiction $f(e_{PQ}) = \mu(e_{PQ})$. Then either $f(e_{PQ})$ lies along $L$
or is symmetric about $L$. 

If $f(e_{PQ})$ lies along $L$, consider either endpoint $p$ of $e_P$. $f(p)$ is
either in the interior or on the boundary of $f(\puz)$. If $f(p)$ is interior,
then the two edges of $P$ incident to $f(p)$ must be connected to two different
pieces, contradicting that $P$ is a leaf. Alternatively, $f(p)$ is on the
boundary, and a vertex of some other piece $P'$ must contain $f(p)$,
contradicting exclusion from Case 1. 

Alternatively $f(e_{PQ})$ is symmetric about $L$. Because $P$ is a leaf, it
connects to the rest of the symmetric assembly only through $f(e_{PQ})$, so for
the assembly to be symmetric, $f(P)$ must be the same as $\mu(P)$, and piece $P$
is a line symmetric polygon.

\begin{figure}
  \centering
  \includegraphics[width=\textwidth]{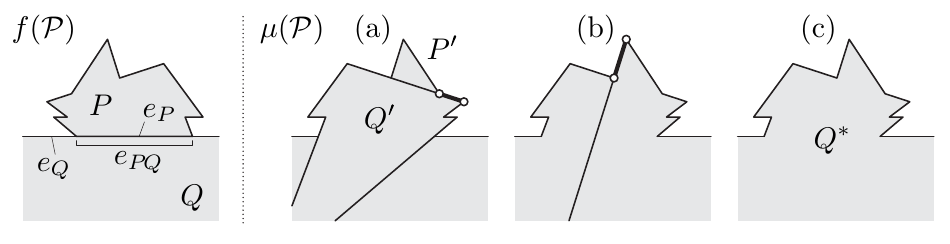}
  \caption{Possible topological configurations of $\mu(P)$.}
  \label{fig:mirrorcase}
\end{figure}

So $f(e_{PQ})$ and $\mu(e_{PQ})$ are not the same line segment. We claim that
$\mu(P)$ is exactly covered by another piece of the SAP congruent to $P$.
Suppose for contradiction it were not. Then by Lemma~\ref{lem:partial}, $P$ must
be fully connected to $Q$. Then $\mu(P)$ either (a) contains a piece as a strict
subset, (b) does not fully contain a piece but intersects the interiors of
multiple pieces, or (c) is a strict subset of a single piece (see
\figurename~\ref{fig:mirrorcase}).

First suppose (a), so $\mu(P)$ contains some piece as a strict subset. We will
say that piece $P$ \emph{covers} a piece $P'$ if $f(P')$ is a strict subset of
$\mu(P)$. We will identify a leaf piece $P'$ covered by $P$, whose parent
connection constructs the length of an edge of $P$, contradicting exclusion from
Case 2. To find such a piece, consider any piece $R$ that is not covered by $P$,
and let $S$ be a piece from among all pieces covered by $P$ that has longest
distance to $R$ in the connection graph. This condition ensures that $S$ is a
leaf, connected to some piece $Q'$ through edge $e_{Q'}$ from $Q'$. Because $S$
is covered by $P$, at least one endpoint of $e_{Q'}$ maps to point contained in
$\mu(P)$. Let $q$ be such an endpoint. Point $f(q)$ is a vertex of the symmetric
assembly or else the connection of $Q$ and some other piece would share a vertex
on their connection at $f(q)$. Let $P'$ be the piece connected to $e_{Q'}$ with
connection closest to $q$. $P'$ is a leaf or else $S$ would not have had a
longest distance to $R$ in the connection graph. Further, because $S$ is covered
by $P$, so is $P'$. By Lemma~\ref{lem:partial}, the connection between $P'$ and
$Q'$ must be fully connected. If $f(e_{P'}) = f(e_{Q'})$ then $P'$ and $Q'$ share
vertices along their connection, contradicting exclusion from Case 1. If
$f(e_{Q'})\subset f(e_{P'})$, then because $P'$ is a leaf, $f(e_{P'})\setminus
f(e_{Q'})$ constructs the lengths of two edges of $P$, contradicting exclusion
from Case 2. So edge $e_{P'}$ fully connects $P'$ to $Q'$ in the assembly. And
because no other piece connects to $e_{Q'}$ between vertex $q$ and the connection
between $P'$ and $Q'$, the distance betwen them constructs the length of an edge
of $P$, contradicting exclusion from Case 2. So $\mu(P)$ does not contain a
piece as a strict subset.

Now suppose (b), so that two connected pieces intersect $\mu(P)$. The edges
connecting these two pieces must overlap in $\mu(P)$ to construct a length equal
to an edge of $P$, contradicting exclusion from Case~2. So $\mu(P)$ does not
intersect the interior of multiple branch pieces.

Finally suppose (c), and let $\mu(P)$ be the strict subset of some piece $Q^*$.
Let $\ell$ be the line collinear with segment $f(e_{PQ})$, and let
$e_{\ell}$ be the subset of $Q$ that maps to the largest connected subset of
$\ell \cap f(Q)$ containing $f(e_{PQ})$. Consider the two disconnected sections
of the boundary of $Q$ between an endpoint of $e_{PQ}$ and an endpoint of
$e_{\ell}$, which must each be more than an isolated point or exclusion from
Case~1 would be violated. Piece $P$ has at most one branch sibling, so at most
one of these sections can be connected to a branch. Let $q$ be an endpoint of
$e_\ell$ in a section not connected to a branch. 

\begin{figure}[htb]
  \centering
  \includegraphics[width=\textwidth]{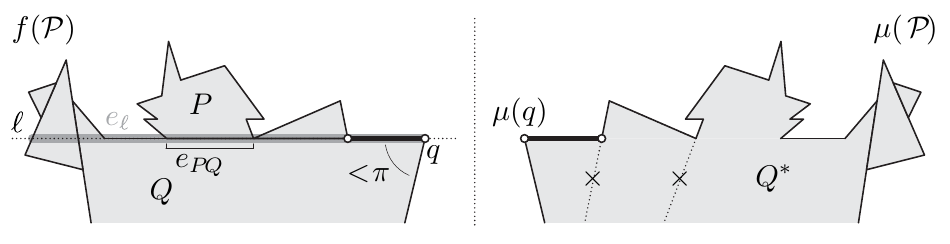} \\
  \includegraphics[width=\textwidth]{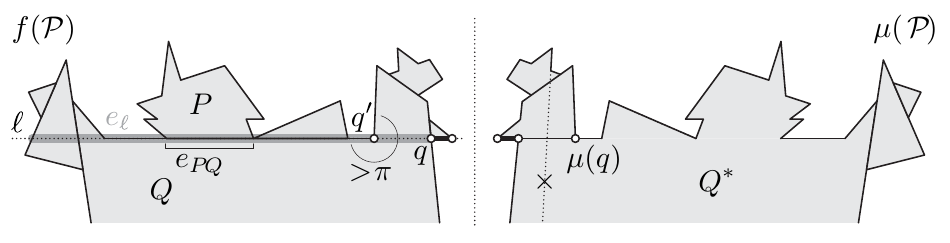}
  \caption{Considering if $\mu(P)$ is a strict subset of $Q^*$ and 
  the boundary between $e_{PQ}$ and $q$ is a [Left] straight line or [Right] not a straight line.}
  \label{fig:convexang}
\end{figure}

Consider the boundary of $Q$ between $e_{PQ}$ and $q$. Suppose this boundary
were a line segment subset of $e_Q$, implying the internal angle of $Q$ at $q$
is less than $\pi$; see \figurename~\ref{fig:convexang}. Point $q$ is not
included in the connection between $Q$ and another piece through $e_Q$. If it
were, it would be a partial connection with a leaf piece, which by
Lemma~\ref{lem:partial} would be part of a partially connected pair
contradicting the minimality of $f(\puz)$. Further, $\mu(q)$ is a vertex of
$f(Q^*)$ or else $Q^*$  would connect to another piece somewhere on the segment
between $\mu(e_{PQ})$ and $\mu(q)$, and their connection would construct an edge
of the same length as an edge from a leaf connected to $f(e_Q)$, contradicting
exclusion from Case~2. The edge of $Q^*$ adjacent to $\mu(q)$ contained in
$\mu(e_Q)$ will have the same length as the subset of $e_Q$ between $q$ and a
vertex of a leaf, contradicting exclusion from Case~2.

Thus, the boundary of $Q$ between $e_{PQ}$ and $q$ is not a line segment, so
$f(Q)$ must cross $\ell$, and the endpoint $q'$ of $e_Q$ in this section is a
vertex of $Q$ with internal angle greater than $\pi$; see
\figurename~\ref{fig:convexang}. By the same argument as we applied to $q$ in
the preceding paragraph, $\mu(q')$ must be in $f(Q^*)$, and if it were a vertex,
we would have the same contradiction as with $q$ before. However this time
$\mu(q')$ need not be a vertex of $f(Q^*)$ because $f(Q^*)$ may extend past
$\mu(q')$, with $Q^*$ connecting to another piece on the other side of $e_\ell$.
However, the connection between these pieces will construct an edge that is the
same length as an edge in either $Q$ or a leaf connected to $Q$, and we have
arrived at our final contradiction. So if $P$ is not line symmetric, $\mu(P)$ is
itself a piece of the SAP congruent to $P$, contradicting the minimality of
$\puz$.
\end{proof}

Because every symmetric assembly can be classified as one of these cases, we can
check for each case to decide whether the SAP has a symmetric assembly. Given a SAP
that does not satisfy Case~1 or Case~2, by Lemma~\ref{lem:cases} it must satisfy
Case~3 if it has a symmetric assembly. However, satisfying Case~3 is not
sufficient to ensure a symmetric assembly. For example, two congruent regular
polygons with many sides and a single regular star with many spikes cannot by
themselves form a symmetric assembly, though they satisfy Case~3, because no
pair of edges can be joined without making the pieces overlap (for example, the
Case 3 example from Figure~10, exchanging to the triangular pieces for large
regular hexagons). Thus given a SAP in Case~3, we must search the configuration
space of possible connected arrangements of the pieces for an arrangement that
forms a simple polygon. 
 
Recall that the connection graph for a symmetric assembly not in Case~1 must be
a tree. For a SAP with $k$ pieces, consisting of at most $n$ vertices in total,
Cayley's formula says the number of distinct connection trees is
$k^{k-2}$~\cite{cayley}. However, even if two pieces are connected, they could
be connected through $O(n^2)$ different pairs of directed edges, so the number of
different \emph{edge distinguishing connection trees}, connection trees
distinguishing between which pairs of edges are connected, can be no more than
$n^{2k}k^{k} = O(n^{2k})$ ($k$ is constant). As an instance of Case~3, $\puz$
consists of one or more symmetric pieces, with the rest being congruent pairs.
Let $\mathcal{D}_\puz$ and $\mathcal{D}'_\puz$ be maximal disjoint subsets of
$\puz$ such that there exists a matching
$\eta:\mathcal{D}_\puz\rightarrow\mathcal{D}'_\puz$ between pieces in
$\mathcal{D}_\puz$ and $\mathcal{D}'_\puz$ such that matched pairs are
congruent. Let $\mathcal{S}_\puz$ be the set of symmetric pieces in $\puz$ not
in $\mathcal{D}_\puz$ or $\mathcal{D}'_\puz$. Let  $\mathcal{D}_s$ denote some
subset of the symmetric pieces contained in $\mathcal{D}_\puz$, and define a
\emph{trunk} to be a subset of symmetric pieces $\mathcal{R} =
\mathcal{S}_\puz\cup\mathcal{D}_s\cup\eta(\mathcal{D}_s)$ that can be connected
into a simple polygon without overlap while aligning each of their lines of
symmetry to a common line $L$ (see \figurename~\ref{fig:trunk}). Define a
\emph{half tree} $T$ to be an edge distinguishing connection tree on
$\mathcal{R}\cup\mathcal{D}_\puz$ such that every piece in $\mathcal{D}_\puz$
connected to a piece $R$ in $\mathcal{R}$ connects through an edge of $R$
intersecting the same half-plane bounded by $L$. We call this half-plane the
\emph{connecting half-plane}, with the other half-plane the \emph{free
half-plane}. The reason we define half trees is if we can find a point in their
configuration space for which pieces do not intersect and for which pieces in
$\mathcal{D}_\puz$ not in the trunk do not intersect the free half-plane, we can
place the remaining congruent pieces in $\mathcal{D}_\puz\setminus
\mathcal{D}_s$ at the mirror image of their respective matched pairs to complete
a symmetric assembly. If a symmetric assembly exists satisfying Case 3, the
assembly will correspond to a point in the constructed configuration space by
definition.

\begin{figure}[tb]
  \centering
  \includegraphics[width=\textwidth]{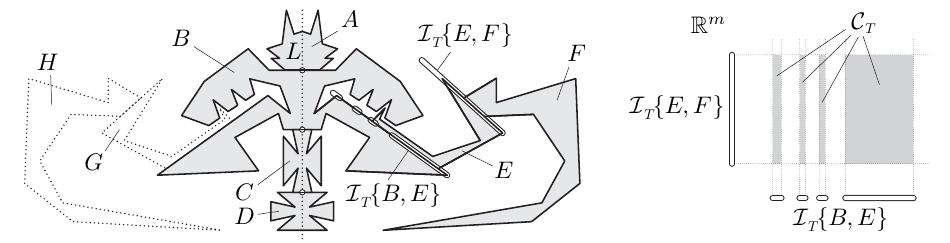}
  \caption{An example showing a SAP $\puz$ satisfying Case 3, with $\mathcal{S}_\puz = \{A,B\}$,
  $\mathcal{D}_\puz = \{C,E,F\}$, $\mathcal{D}'_\puz = \{D,G,H\}$,
  $\mathcal{D}_s = \{C\}$, $\eta(\mathcal{D}_s) = \{D\}$, and
  trunk $\mathcal{R} = \{A,B,C,D\}$. $\mathcal{I}_T$ for two connected pieces
  in the trunk is just a single point as shown by the midpoint of their
  connection. Pieces not in the trunk have a degree of freedom sliding along
  their connection. $\mathcal{I}_T\{E,F\}$ is a single interval where $F$ can
  attach to $E$, while $\mathcal{I}_T\{B,E\}$ is four intervals. The right
  diagram shows $\mathcal{C}_T$ the Cartesian product of each $\mathcal{I}_T$. }
  \label{fig:trunk}
\end{figure}

Let $\mathcal{T}_\puz$ be the set of possible half trees of $\puz$. Let
$\mathcal{L}_T$ be the set of undirected edges $\{P,Q\}$ where piece $P$ is
connected to piece $Q$ in tree $T\in\mathcal{T}_\puz$, and let $m =
|\mathcal{L}_T| \leq k$. For a fixed edge distinguishing connection tree, the
orientation of each piece is fixed as pieces may only translate along their
specified connection. We want to define a set of intervals
$\mathcal{I}_T\{P,Q\}$ where we could join pieces $P$ and $Q$ along respective
edges $e_P$ to $e_Q$ that are connected in tree $T$, while together forming a
simple polygon without overlap. For each $\{P,Q\}\in\mathcal{L}_T$ with $e_P$
and $e_Q$ the respective connecting edges of $P$ and $Q$ with $\lambda(e_P) \geq
\lambda(e_Q)$, let $\mathcal{I}_T\{P,Q\}$ be defined as follows. If $P$ and $Q$
are both in $\mathcal{R}$, let $\mathcal{I}_T\{P,Q\}$ be the empty set if
$\join(e_P,e_Q,d_{PQ})$ is the empty set and $\{d_{PQ}\}$ otherwise, where we
use $d_{PQ}$ to denote $|\lambda(e_P)-\lambda(e_Q)|/2$, the distance $d$ would
need to be in order to align the midpoints of $e_P$ and $e_Q$. Alternatively if
$P$ or $Q$ are not both in $\mathcal{R}$, let $\mathcal{I}_T\{P,Q\}$ be the
closure of the set of distances $d$ for which $\join(e_P,e_Q,d)$ is a simple
polygon for which $P$ and $Q$ do not share a vertex along their connection. Note
that if $P$ or $Q$ are not both in $\mathcal{R}$, $\mathcal{I}_T\{P,Q\}$ will be
a sequence of positive length closed intervals. Each interval endpoint
represents a point at which $P$ and $Q$ would just intersect, no longer forming
a simple polygon. The number of such points is upper bounded by the $O(n^2)$
possible intersections between some edge of $P$ and some edge of $Q$ when
sliding the two pieces along their connection; so the number of distinct
intervals in $\mathcal{I}_T\{P,Q\}$ is at most quadratic in the number of
vertices, $O(n^2)$. Any fixed arrangement of the pieces consistent with an edge
distinguishing connection tree $T$ joins each pair of pieces by fixing one point
in every $\mathcal{I}_T\{P,Q\}$, so the set of configurations is a subset of
$\mathbb{R}^m$. Ignoring overlap between pieces that are not connected, the
configuration space $\mathcal{C}_T$ of possible arrangements is defined as the
Cartesian product of $\mathcal{I}_T\{P,Q\}$ for every $\{P,Q\}\in\mathcal{L}_T$.
Thus $\mathcal{C}_T$ is a set of $O(n^m)$ disjoint $m$-dimensional
hyperrectangles in $\mathbb{R}^m$.

We now describe the subset of $\mathbb{R}^m$ where intersection occurs between
two pieces that are not connected in $T$. If two pieces in a configuration
overlap, by continuity there exist two edges $e_P$ and $e_Q$ from two distinct
pieces $P$ and $Q$ that also intersect. The positions of $e_P$ and $e_Q$ are
translations parameterized by a point in $\mathcal{C}_T$ and the region in which
the two edges intersect is a convex region
$\mathcal{X}_T\{e_P,e_Q\}\subset\mathbb{R}^m$ bounded by four hyperplanes
forming the $m$-dimensional parallelogram representing the intersection of the
two edges. For each $O(n^2)$ pair of edges from distinct pieces that are not
connected, we can subtract each $\mathcal{X}_T\{e_P,e_Q\}$ from $\mathcal{C}_T$
to form $\mathcal{C}'_T$. 

If $\mathcal{C}'_T$ is empty, there will certainly be no symmetric assembly
satisfying Case 3. If $\mathcal{C}'_T$ is a single point, tree $T$ places all
pieces in the trunk $\mathcal{R}$ to form a symmetric assembly. Lastly, if
$\mathcal{C}'_T$ is non-empty and contains a point in its interior, then there
exists a symmetric assembly because it will be a point in the configuration space
avoiding overlap between pieces. Points on the boundary of $\mathcal{C}'_T$
correspond to configurations that are non-simple (the symmetric assembly is not
homeomorphic to a disc), as the boundaries of each $\mathcal{I}_T$ not between
two pieces in $\mathcal{R}$ and the boundaries of each $\mathcal{X}_T$
correspond to configurations which produces a hole in the assembly or a cycle in
the connection graph. Thus, if $\puz$ has a symmetric assembly satisfying Case
3, $\mathcal{C}'_T$ will have a point on its interior or be a single point.

Consider the function \three described in Algorithm~1.

\begin{algorithm}[t]
  \DontPrintSemicolon
  \myfun{\three{$\puz$}}{
  \Input{Symmetric assembly puzzle $\puz$.}
  \Output{\True if $\puz$ has a Case 3 symmetric assembly, // \False otherwise.}
  \For{$T\in\mathcal{T}_\puz$}{
    $\mathcal{C}'_T \leftarrow \mathcal{C}_T$\;
    \For{$\{P,Q\}\in \mathcal{L}_T$}{
      $\mathcal{C}'_T \leftarrow \mathcal{C}'_T\setminus \mathcal{X}_T\{e_P,e_Q\}$\;
    }
    \If{$\operatorname{interior}(\mathcal{C}'_T) \neq \emptyset$}{
     \KwRet \True\;
    }
    \ElseIf{$\mathcal{C}'_T \neq \emptyset$ \textbf{and} $\mathcal{D}_\puz = \emptyset$}{
     \KwRet \True\;
    }
  }
  \KwRet \False\;
  }{}{}
  \label{alg:case3}
  \caption{Pseudocode for function \textsc{hasAssemblyCase3}($\puz$)}
\end{algorithm}  

\begin{lemma}
Given symmetric assembly puzzle $\puz$ that satisfies Case 3,
function \three{$\puz$} returns \True if and only if $\puz$ has a symmetric
assembly, and terminates in $O(n^{6k})$ time.
\label{lem:case3}
\end{lemma}

\begin{proof}

We can test all pieces for line symmetry or congruence in $O(nk)$
time~\cite{wolter1985optimal}. If $\puz$ has a symmetric assembly satisfying
Case~3 with nonempty $\mathcal{D}_\puz$, $\mathcal{C}'_T$ will have a point on
its interior for some tree $T$ as argued above; or if $\mathcal{D}_\puz$ is
empty, $\mathcal{C}'_T$ must be nonempty, i.e., a single point corresponding to
constructing a trunk from all the pieces. There are $O(n^{2k})$ elements of
$\mathcal{T}_\puz$. There are $m = O(k)$ interval sets $\mathcal{I}_T\{P,Q\}$
each having computational complexity $O(n^2)$, so we can construct
$\mathcal{C}_T$ naively in $O(n^{2k})$ time. The union $X_T$ of the $O(n^2)$
regions $\mathcal{X}_T\{e_P,e_Q\}$, which are $m$-dimensional convex regions,
has computational complexity at most $O(n^{2m})$, so the final computational
complexity of $\mathcal{C}'_T=\mathcal{C}_T\setminus X_T$ is at most $O(n^{4m})$
and can be computed in as much time. Checking each of the $O(n^{2k})$ elements
of $\mathcal{T}_\puz$ in this way yields the running time for \three bounded by
$O(n^{6k})$. \end{proof}

Our brute force algorithm 
 \hasSA{$\puz$} is described in Algorithm~2. 
 
\begin{algorithm}[htb]
  \DontPrintSemicolon
  \myfun{\hasSA{$\puz$}}{
  \Input{Symmetric assembly puzzle $\puz$.}
  \Output{\True if $\puz$ satisfies Case 1 or Case 2 or Case 3, \False otherwise.}
  \For{$e_P\in E_P,e_Q\in E_Q, \{P,Q\}\subset \puz$}{
    $S \leftarrow$ \join{$e_P,e_Q,0$}\;
    $\puz' \leftarrow (\puz\setminus \{P, Q\})\cup \{S\}$\;
    \uIf{$S\neq\emptyset$ \textbf{and} \hasSA{$\puz'$}}{
      \KwRet \True\tcp*{Case 1}
    }
    \For{$e_R\in E_R, R\in \puz$}{
      \If{$\lambda(e_R)< \lambda(e_P)$}{
        $S \leftarrow$ \join{$e_P,e_Q,\lambda(e_R)$}\;
        $\puz' \leftarrow (\puz\setminus \{P, Q\})\cup \{S\}$\;
        \uIf{$S\neq\emptyset$ \textbf{and} \hasSA{$\puz'$}}{
          \KwRet \True\tcp*{Case 2}
        }
      }
    }
  }
  \KwRet \three{$\puz$}\tcp*{Case 3}
  }{}{}
  \label{alg:full}
  \caption{Pseudocode for function \textsc{hasAssembly}($\puz$)}
\end{algorithm} 

\begin{lemma}
Function \hasSA{$\puz$} returns \True if and only if $\puz$ has a symmetric
assembly that satisfies either
Case 1, Case 2, or Case 3, and terminates in $O(n^{6k})$ time.
\label{lem:alg}
\end{lemma}

\begin{proof}
We prove by induction.
For the base case, $\puz$ consists of only a single piece satisfying Case 3,
which will drop directly to the last line of the algorithm checking Case 3 which, 
by Lemma~\ref{lem:case3} will evaluate correctly.
Now suppose \hasSA returns a correct evaluation for SAPs containing $k-1$ pieces.
Then we show \hasSA returns a correct evaluation for SAPs containing $k$ pieces.

The outer \textbf{for} loop of \hasSA cycles through every pair of directed 
edges $e_P=(p_1,p_2)$ and $e_Q=(q_1,q_2)$ taken from different pieces $P$ and $Q$.
For each pair, \hasSA first checks to see if there exists a symmetric assembly for
which $e_P$ is connected to $e_Q$ with $p_1$ coincident to $q_1$, which would
satisfy Case 1. If one exists, then joining $P$ and $Q$ 
into one piece as described would produce a SAP $\puz'$ with one fewer piece that also has a symmetric assembly.
Then evaluating \hasSA on the smaller instance will return correctly
by induction. Because the
outer \textbf{for} loop checks every possible pair of edges that could be joined in a symmetric assembly
satisfying Case 1, \hasSA will return \True if $\puz$ satisfies Case 1.

Next \hasSA checks to see if there exists a symmetric assembly for
which $e_P$ is connected to $e_Q$ with $p_1$ and $q_1$ separated by a distance equal to
the length of some other edge $e_R$ in $\puz$, which would satisfy Case 2. 
In the same way as with Case 1, both \textbf{for}
loops check every possible pair of edges and that could be joined 
at every possible length that could produce a symmetric assembly
satisfying Case 2, so \hasSA will return \True if $\puz$ satisfies Case 2.

Otherwise, no symmetric assembly exists satisfying Case 1 or Case 2. 
By Lemma~\ref{lem:case3}, \three correctly evaluates if $\puz$ is in Case 3,
so \hasSA returns a correct evaluation for SAPs containing $k$ pieces.

Let $T(k)$ be the running time of \hasSA on an instance with $k$ pieces.
Then the recurrence relation for \hasSA is 
$T(k) = O(n^3)T(k-1)+O(n^{6k})$,
where $O(n^{6k})$ is the running time given by Lemma~\ref{lem:case3}.
Running time for Case 3 at the leaves dominates the recurrence
relation so \hasSA terminates in $O(n^{6k})$.
\end{proof}

Now we can determine whether a symmetric assembly puzzle
with a constant number of pieces has a symmetric assembly in polynomial time.

\begin{proof} [of \theoremname~\ref{th:constpieces}] By Lemma~\ref{lem:cases},
if the SAP has a symmetric assembly, it satisfies either Case 1, Case 2, or Case
3, and by Lemma~\ref{lem:alg} \hasSA{$\puz$} can correctly determine if it has a
symmetric assembly satisfying one of the cases in polynomial time, proving the
claim. \end{proof}

\section{Conclusion}

Several open questions remain.  It may be interesting to consider SAPs
for special classes of shapes $P_i\in \puz$.  We conjecture that SAPs
remain hard for instances in which the shapes $P_i$ are right
triangles (Conjecture~\ref{conj:TriangleNPComplete}).
Are SAPs hard for a constant number $k=O(1)$ of pieces if
the target shape is allowed to be nonsimple (a polygon with holes)?
Are SAPs fixed-parameter tractable with respect to the number $k$ of
pieces? (We conjecture W[1]-hardness.)

\medskip

\noindent \textbf{Acknowledgements:} Many of the authors were introduced to
symmetric assembly puzzles during the 30th Winter Workshop on Computational
Geometry at the Bellairs Research Institute of McGill University, March 2015.
Korman was supported by MEXT KAKENHI No.~17K12635 and the NSF award CCF-1422311. Mitchell is supported in part by the National Science Foundation
(CCF-1526406). Van Renssen and Roeloffzen were supported by JST ERATO Grant
Number JPMJER1201, Japan. Uno is supported in part by JSPS KAKENHI No.~17K00017 and by JST CREST Grant Number JPMJCR1402, Japan.

\bibliographystyle{abbrv}
\bibliography{lsp}

\end{document}